\documentclass[a4paper,UKenglish]{article}

\usepackage{a4wide}
\usepackage{color}
\usepackage{amsmath}
\usepackage{amsfonts}
\usepackage{amsthm}
\usepackage{tikz}
\usetikzlibrary{decorations.pathreplacing}
\usepackage[rightcaption]{sidecap}
\sidecaptionvpos{figure}{l}
\usepackage{wrapfig}

\newtheorem{theorem}{Theorem}
\newtheorem{lemma}[theorem]{Lemma}
\newtheorem{corollary}[theorem]{Corollary}
\newtheorem{definition}[theorem]{Definition}
\newtheorem{main-result}{Main Result}
\newtheorem{claim}[theorem]{Claim}
\newtheorem{observation}{Observation}

\renewcommand{\epsilon}{\varepsilon}

\newcommand{\OPT}{\mbox{OPT}}
\newcommand{\SW}{\mbox{SW}}

\newcommand{\RR}{\mathbb{R}}
\newcommand{\Var}{\mathrm{Var}}

\newif\ifnotes
\notesfalse
\notestrue
\ifnotes
\newcommand{\bojana}[1]{\textcolor{blue}{{\footnotesize #1}}\marginpar{\raggedright\tiny \textcolor{purple}{Bojana}}}

\definecolor{green}{rgb}{0,1,0}
\newcommand{\thomas}[1]{\textcolor{red}{{\footnotesize #1}}\marginpar{\raggedright\tiny \textcolor{red}{Thomas}}}

\else
\newcommand{\bojana}[1]{}
\newcommand{\thomas}[1]{}
\fi

\newcommand{\tk}[1]{\textcolor{black}{#1}}
\newcommand{\bk}[1]{\textcolor{black}{#1}}

\makeatletter
\DeclareRobustCommand{\rvdots}{%
  \vbox{
    \baselineskip4\p@\lineskiplimit\z@
    \kern-\p@
    \hbox{.}\hbox{.}\hbox{.}
  }}
\makeatother

\usepackage{microtype}

\title{Price of Anarchy for Mechanisms with Risk-Averse Agents
\footnote{This work is supported by DFG through Cluster of Excellence MMCI.}
}

\author{Thomas Kesselheim\thanks{University of Bonn, Institute of Computer Science, Bonn, Germany {\tt thomas.kesselheim@uni-bonn.de}} \and
Bojana Kodric\thanks{MPI for Informatics and Saarland University, Saarbr\"ucken, Germany {\tt bojana@mpi-inf.mpg.de}}}

\date{}

\begin{document}

\maketitle

\begin{abstract}
We study the price of anarchy of mechanisms in the presence of risk-averse agents. Previous work has focused on agents with quasilinear utilities, possibly with a budget. Our model subsumes this as a special case but also captures that agents might be less sensitive to payments than in the risk-neutral model. We show that many positive price-of-anarchy results proved in the smoothness framework continue to hold in the more general risk-averse setting.
A sufficient condition is that agents can never end up with negative quasilinear utility after playing an undominated strategy.
This is true, e.g., for first-price and second-price auctions. For all-pay auctions, similar results do not hold: We show that there are Bayes-Nash equilibria with arbitrarily bad social welfare compared to the optimum.
 \end{abstract}


\section{Introduction}\label{sec:introduction}
Many practical, ``simple'' auction mechanisms are not incentive compatible, making it beneficial for agents to behave strategically. A standard example is the first-price auction, in which one item is sold to one of $n$ agents. Each of these agents is asked to report a valuation; the item is given to the agent reporting the highest value, who then has to pay what he/she reported.
A common way to understand the effects of strategic behavior is to study resulting equilibria and to bound the \emph{price of anarchy}. That is, one compares the social welfare that is achieved at the (worst) equilibrium of the induced game to the maximum possible welfare. Typical equilibrium concepts are Bayes-Nash equilibria and (coarse) correlated equilibria, which extend mixed Nash equilibria toward incomplete information or learning settings respectively.
A key assumption in these analyses is that agents are \emph{risk neutral}: Agents are assumed to maximize their expected quasilinear utility, which is defined to be the difference of the value associated to the outcome and payment imposed to the agent. So, an agent having a value of $1$ for an item would be indifferent between getting this item with probability $10 \%$ for free and getting it all the time, paying $0.9$.

However, there are many reasons to believe that agents are not risk neutral. For instance, in the above example the agent might favor the certain outcome to the uncertain one. Therefore, in this paper, we ask the question: \emph{What ``simple'' auction mechanisms preserve good performance guarantees in the presence of risk-averse agents?}

The standard model of risk aversion in economics (see, e.g., \cite{mas1995microeconomic}) 
is to apply a (weakly) concave function to the quasilinear term. That is, if agent $i$'s outcome is $x_i$ and his payment is $p_i$, his utility is given as $u_i(x_i, p_i) = h_i(v_i(x_i) - p_i)$, where $h_i\colon \RR \to \RR$ is a weakly concave, monotone function. That is, for $y,y'\in\mathbb{R}$ and for all $\lambda \in [0, 1]$, it holds that $h_i(\lambda y + (1-\lambda) y') \geq \lambda h_i(y) + (1-\lambda) h_i(y')$. Agent $i$ is risk neutral if and only if $h_i$ is a linear function. If the function is strictly concave, this has the effect that, by Jensen's inequality, the utility for fixed $x_i$ and $p_i$ is higher than for a randomized $x_i$ and $p_i$ with the same expected $v_i(x_i) - p_i$.

We compare outcomes based on their social welfare, which is defined to be the sum of utilities of all involved parties including the auctioneer. That is, it is the sum of agents' utilities and their payments $\SW(\mathbf{x}, \mathbf{p}) = \sum_i u_i(x_i, p_i) + \sum_i p_i$. In the quasilinear setting this definition of social welfare coincides with the sum of values $\sum_i v_i(x_i)$. With risk-averse utilities they usually differ. However, all our results bound the sum of values and therefore also hold for this benchmark. 

We assume that the mechanisms are oblivious to the $h_i$-functions and work like in the quasilinear model. Only the individual agent's perception changes. This makes it necessary to normalize the $h_i$-functions because otherwise they could be on different scales, e.g., if $h_1(y) = y$ and $h_2(y) = 1000 \cdot y$, which would be impossible for the mechanism to cope with without additional information. Therefore, we will assume that $u_i(\mathbf{x},p_i) = v_i(\mathbf{x})$ if $p_i = 0$ and that $u_i(\mathbf{x},p_i) = 0$ if $p_i = v_i$. That is, for the two cases that $p_i$ is either $0$ or the full value, the utility matches exactly the quasilinear one. However, due to risk aversion, the agents might be less sensitive to payments. \footnote{We note here that this will not in turn allow the mechanism to arrive at huge utility gains, as compared to the quasilinear model, for example, by increasing payments arbitrarily.
Indeed, Lemma~\ref{lemma:OPT_relation} in Section~\ref{sec:model} will show that the difference between the two optima is bounded by at most a multiplicative factor of $2$.}


\subsection{Our Contribution}

We give bounds on the price of anarchy for Bayes-Nash and (coarse) correlated equilibria of mechanisms in the presence of risk-averse agents. Our positive results are stated within the smoothness framework, which was introduced by \cite{Roughgarden09}. We use the version that is tailored to quasilinear utilities by \cite{SyrgkanisT13}, which we extend to mechanism settings with general utilities (for a formal definition see Section~\ref{sec:general_smoothness}). Our main positive result states that the loss of performance compared to the quasilinear setting is bounded by a constant if a slightly stronger smoothness condition is fulfilled.

\begin{main-result}
\label{main-result-smooth}
\label{MAIN-RESULT-SMOOTH}
Given a mechanism with price of anarchy $\alpha$ in the quasilinear model provable via smoothness such that the deviation guarantees non-negative utility, then this mechanism has price of anarchy at most $2 \alpha$ in the risk-averse model.
\end{main-result}

This result relies on the fact that the deviation action to establish smoothness guarantees agents non-negative utility. A sufficient condition is that all undominated strategies never have negative utility. First-price and second-price auctions satisfy this condition, we thus get constant price-of-anarchy bounds for both of these auction formats.

In an all-pay auction every positive bid can lead to negative utility. Therefore, the positive result does not apply. As a matter of fact, this is not a coincidence because, as we show, equilibria can be arbitrarily bad.

\begin{main-result}
\label{main-result-allpay}
\label{MAIN-RESULT-ALLPAY}
The single-item all-pay auction has unbounded price of anarchy for Bayes-Nash equilibria, even with only three agents.
\end{main-result}


This means that although equilibria of first-price and all-pay auctions have very similar properties with quasilinear utilities, in the risk-averse setting they differ a lot. We feel that this to some extent matches the intuition that agents should be more reluctant to participate in an all-pay auction compared to a first-price auction.

In our construction for proving Main Result~\ref{main-result-allpay}, we give a symmetric Bayes-Nash equilibrium for two agents. The equilibrium is designed in such a way that a third agent of much higher value would lose with some probability with every possible bid. Losing in an all-pay auction means that the agent has to pay without getting anything, resulting in negative utility. In the quasilinear setting, this negative contribution to the utility would be compensated by respective positive amounts when winning. For the risk-averse agent in our example, this is not true. Because of the risk of negative utility, he prefers to opt out of the auction entirely.

We also consider a different model of aversion to uncertainty, in which solution concepts are modified. Instead of evaluating a distribution over utilities in terms of their expectation, agents evaluate them based on the expectation minus a second-order term. We find that this model has entirely different consequences on the price of anarchy. For example, the all-pay auction has a constant price of anarchy in correlated and Bayes-Nash equilibria, whereas the second-price auction can have an unbounded price of anarchy in correlated equilibria.

\subsection{Related Work}\label{subsec:related}


Studying the impact of risk-averseness is a regularly reoccurring theme in the literature. A proposal to distinguish between money and the utility of money, and to model risk aversion by a utility function that is concave first appeared in~\cite{10.2307/1909829}. The expected utility theory, which basically states that the agent's behavior can be captured by a utility function and the agent behaves as a maximizer of the expectation of his utility, was postulated in~\cite{zbMATH03106184}. This theory does not capture models that are standardly used in portfolio theory, ``expectation minus variance'' or ``expectation minus standard deviation''~\cite{markowitz1968portfolio}, the latter of which we also consider in Section~\ref{sec:other-models}.

In the context of mechanisms, one usually models risk aversion by concave utility functions. One research direction in this area is to understand the effects of risk aversion on a given mechanism. For example, \cite{DBLP:journals/ijgt/FibichGS06} studies symmetric equilibria in all-pay auctions with a homogenous population of risk-averse players. \tk{Due to symmetry and homogeneity, in this case, equilibria are fully efficient, that is, the price of anarchy is $1$.}
In~\cite{mathews2006} a similar analysis for auctions with a buyout option is performed; \cite{DBLP:conf/wine/HoyIL16} considers customers with heterogeneous risk attitudes in mechanisms for cloud resources. In~\cite{DBLP:conf/sigecom/DuttingKT14} it is shown that for certain classes of mechanisms the correlated equilibrium is unique and has a specific structure depending on the respective valuations but independent of the actual utility function. One consequence of this result is that risk aversion does not influence the allocation outcome or the revenue.

Another direction is to design mechanisms for the risk-averse setting. For example, the optimal revenue is higher because buyers are less sensitive to payments. In a number of papers, mechanisms for revenue maximization are proposed \cite{matthews1983,maskin1984optimal,DBLP:conf/sigecom/SundararajanY10,hu2010risk,bhalgat2012mechanism,fu2013prior}. Furthermore, randomized mechanisms that are \emph{truthful in expectation} lose their incentive properties if agents are not risk neutral. Black-box transformations from truthful-in-expectation mechanisms into ones that fulfill stronger properties are given in~\cite{DBLP:journals/corr/abs-1206-2957} and~\cite{DBLP:journals/teco/HoeferKV16}.

Studying the effects of risk aversion also has a long history in game theory, where different models of agents' attitudes towards risk are analyzed. One major question is, for example, if equilibria still exist and if they can be computed \cite{DBLP:conf/sagt/FiatP10,hoy2012concavity}. Price of anarchy analyses have so far only been carried out for congestion games. Tight bounds on the price of anarchy for atomic congestion games with affine cost functions under a range of risk-averse decision models are given in~\cite{DBLP:journals/teco/PiliourasNS16}.

The smoothness framework was introduced by \cite{Roughgarden09}. Among others, \cite{SyrgkanisT13} tailored it to the quasilinear case of mechanisms. It is important to remark here that our approach is different from the one taken by \cite{DBLP:journals/sigmetrics/MeirP15}. They use the smoothness framework to prove generalized price of anarchy bounds for \bk{nonatomic congestion} games in which players have biased utility functions. They assume that players are playing the ``wrong game'' and their point of comparison is the ``true'' optimal social welfare, \bk{meaning that the biases only determine the equilibira but do not affect the social welfare.} We take the utility functions as they are, including the risk aversion, to evaluate social welfare \bk{in equilibria and also to determine the optimum, which makes our models incomparable.}

For precise relation of von Neumann-Morgenstern preferences to mean-variance preferences, see for instance~\cite{DBLP:journals/eor/Markowitz14}. Mean-variance preferences were explored for congestion games in~\cite{DBLP:journals/ior/NikolovaM14,DBLP:conf/sigecom/NikolovaS15}, while~\cite{klose2014auctioning} studies the bidding behavior in an all-pay auction depending on the level of variance-averseness.

\section{Preliminaries}\label{sec:preliminaries}
\subsection{Setting}
We consider the following setting: There is a set $N$ of $n$ players and $\mathcal{X}$ is the set of possible outcomes. Each player $i$ has a utility function $u_i^{\theta_i}$, which is parameterized by her type $\theta_i \in \Theta_i$. Given a type $\theta_i$, an outcome $\mathbf{x} \in \mathcal{X}$, and a payment $p_i \geq 0$, her utility is $u_i^{\theta_i}(\mathbf{x}, p_i)$. The traditionally most studied case are quasilinear utilities, in which types are valuation functions $v_i\in\mathcal{V}_i$, $v_i\colon \mathcal{X} \to \mathbb{R}$ and $u_i^{v_i}(\mathbf{x},p_i)=v_i(\mathbf{x})-p_i$. Throughout this paper, we will refer to quasilinear utilities by $\hat{u}_i^{v_i}$.

For fixed utility functions and types, the social welfare of an outcome $\mathbf{x} \in \mathcal{X}$ and payments $(p_i)_{i \in N}$ is defined as $\SW^{\boldsymbol{\theta}}(\mathbf{x},\mathbf{p}):=\sum_{i \in N} u_i^{\theta_i}(\mathbf{x}, p_i) + \sum_{i \in N} p_i$. In the quasilinear case, this simplifies to $\sum_{i \in N} v_i(\mathbf{x})$. Unless noted otherwise, by $\OPT(\boldsymbol\theta)$, we will refer to the optimal social welfare under type profile $\boldsymbol\theta$, i.e., $\max_{\mathbf{x}, \mathbf{p}} \SW^{\boldsymbol\theta}(\mathbf{x}, \mathbf{p})$.

A \emph{mechanism} is a triple $(\mathcal{A}, X, P)$, where for each player $i$, there is a set of actions $\mathcal{A}_i$ and $\mathcal{A}=\times_i\mathcal{A}_i$ is the set of action profiles,
$X\colon\mathcal{A}\to \mathcal{X}$ is an allocation function that maps actions to outcomes and $P\colon\mathcal{A}\to \mathbb{R}_+^n$ is a payment function that maps actions to payments $p_i$ for each player $i$. Given an action profile $\mathbf{a} \in \mathcal{A}$, we will use the short-hand notation $u_i^{\theta_i}(\mathbf{a})$ to denote $u_i^{\theta_i}(X(\mathbf{a}), p_i)$.

\subsection{Solution Concepts}
In the setting of \emph{complete information}, the type profile $\boldsymbol\theta$ is fixed. We consider (coarse) correlated equilibria, which generalize Nash equilibria and are the outcome of (no-regret) learning dynamics. A \emph{correlated equilibrium (CE)} is a distribution $\bf a$ over action profiles from $\mathcal{A}$ such that for every player $i$ and every strategy $a_i$ in the support of $\mathbf{a}$ and every action $a_i'\in \mathcal{A}_i$, player $i$ does not benefit from switching to $a_i'$ whenever he was playing $a_i$. Formally,
\[
\mathbb{E}_{\mathbf{a}_{-i}\vert a_i}[u_i(\mathbf{a})]\ge \mathbb{E}_{\mathbf{a}_{-i}\vert a_i}[u_i(a_i', \mathbf{a}_{-i})], \forall a_i'\in\mathcal{A}_i, \forall i\enspace.
\]

In \emph{incomplete information}, the type of each player is drawn from a distribution $F_i$ over her type space $\Theta_i$. The distributions are common knowledge and the draws are independent among players. The solution concept we consider in this setting is the \emph{Bayes-Nash Equilibrium}. Here, the strategy of each player is now a (possibly randomized) function $s_i\colon \Theta_i \to \mathcal{A}_i$. The equilibrium is a distribution over these functions $s_i$ such that each player maximizes her expected utility conditional on her private information. Formally,
\[
\mathbb{E}_{\boldsymbol\theta_{-i}\vert \theta_i}[u_i^{\theta_i}(\mathbf{s}(\boldsymbol\theta))] \ge \mathbb{E}_{\boldsymbol\theta_{-i}\vert \theta_i}[u_i^{\theta_i}(a_i,\mathbf{s}_{-i}(\boldsymbol\theta_{-i}))], \forall a_i\in \mathcal{A}_i, \forall \theta_i\in \Theta_i, \forall i\enspace.
\]

The measure of efficiency is the expected social welfare over the types of the players: given a strategy profile $\mathbf{s}\colon\times_i \Theta_i \to \times_i \mathcal{A}_i$, we consider $\mathbb{E}_{\boldsymbol\theta}[\SW^{\boldsymbol\theta}(\mathbf{s}(\boldsymbol\theta))]$. We compare the efficiency of our solution concept with respect to the expected optimal social welfare $\mathbb{E}_{\boldsymbol\theta}[\OPT(\boldsymbol\theta)]$.

The \emph{price of anarchy (PoA)} with respect to an equilibrium concept is the worst possible ratio between the optimal expected welfare and the expected welfare at equilibrium, that is
\[
\mathrm{PoA} = \max_F \max_{D \in EQ(F)}\frac{\mathbb{E}_{\boldsymbol\theta \sim F}[\OPT(\boldsymbol\theta)]}{\mathbb{E}_{\boldsymbol\theta \sim F, \mathbf{a} \sim D}[\SW^{\boldsymbol\theta}(\mathbf{a})]}\enspace,
\]
where by $F=F_1\times\dots\times F_n$ we denote the product distribution of the players' type distributions and by $EQ(F)$ the set of all equilibria, which are probability distributions over action profiles.

We assume that players always have the possibility of not participating, hence any rational outcome has non-negative utility in expectation over the non-available information and the randomness of other players and the mechanism.

\section{Modeling Risk Aversion}\label{sec:model}
When modeling risk aversion, one wants to capture the fact that a random payoff (lottery) $X$ is less preferred than a deterministic one of value $\mathbf{E}[X]$. The standard approach is, therefore, to apply a concave non-decreasing function $h\colon \mathbb{R} \to \mathbb{R}$ to X and consider $h(X)$ instead. By Jensen's inequality, we now know $\mathbf{E}[h(X)] \leq h(\mathbf{E}[X])$.

In the case of mechanism design, the utility of a risk-neutral agent is defined as the quasilinear utility $v_i(\mathbf{x}) - p_i$. That is, if an agent has a value of $1$ for an item and has to pay $0.9$ for it, then the resulting utility is $0.1$. The expected utility is identical if the agent only gets the item with probability $0.1$ for free. To capture the effect that the agent prefers the certain outcome to the uncertain one, we again apply a concave function $h_i\colon \mathbb{R} \to \mathbb{R}$ to the quasilinear term $v_i(\mathbf{x}) - p_i$. We then consider utility functions $u_i(\mathbf{x}, p_i) = h_i(v_i(\mathbf{x}) - p_i)$ in the setting described in Section~\ref{sec:preliminaries}. Note that the mechanisms we consider do not know the $h_i$-functions. They work as if all utility functions were quasilinear.


We want to compare outcomes based on their social welfare. We use the definition of social welfare being the sum of utilities of all involved parties including the auctioneer. That is, $\SW(\mathbf{x}, \mathbf{p}) = \sum_{i \in N} u_i(\mathbf{x}, p_i) + \sum_{i \in N} p_i$. It is impossible for any mechanism to choose good outcomes for this benchmark if the $h_i$-function are arbitrary and unknown. Therefore, we assume that utility functions are normalized so that the utility matches the quasilinear one for $p_i = 0$ and $p_i = v_i(\mathbf{x})$ (see Figure~\ref{fig:risk_averse_utility}). In more detail, we assume the following \emph{normalized risk-averse utilities}:
\begin{enumerate}
\item \label{prop:monotonicity} $u_i^{v_i}(\mathbf{x}, p_i) \geq u_i^{v_i}(\mathbf{x}, p_i')$ if $p_i \leq p_i'$ (monotonicity)
\item \label{prop:first} $u_i^{v_i}(\mathbf{x}, p_i)=0$ if $p_i=v_i(\mathbf{x})$ (normalization at $p_i=v_i(\mathbf{x})$)
\item \label{prop:second} $u_i^{v_i}(\mathbf{x},p_i)=v_i(\mathbf{x})$ if $p_i=0$ (normalization at $p_i=0$)
\item \label{prop:third} $u_i^{v_i}(\mathbf{x},p_i)\geq v_i(\mathbf{x})-p_i$ if $0\leq p_i\leq v_i(\mathbf{x})$ and $u_i^{v_i}(\mathbf{x},p_i)\leq v_i(\mathbf{x})-p_i$ otherwise (relaxed concavity)
\end{enumerate}

Assumption~\ref{prop:third} is a relaxed version of concavity that suffices our needs for the positive results. Our negative results always fulfill concavity.

\begin{SCfigure}
\centering
\begin{tikzpicture}[yscale=0.5]
\draw[->] (-1,0) -- (4.2,0) node[anchor=north] {$p_i$};
      \draw[->] (0,-1) -- (0,4.2) node[anchor=east] {$u_i(\mathbf{x}, p_i)$};
      \draw[scale=0.5, font=\footnotesize] (6,0) node[anchor=north east] {$v_i(\mathbf{x})$};
      \draw[scale=0.5, font=\footnotesize] (0,6) node[anchor=east] {$v_i(\mathbf{x})$};
      \draw[scale=0.5, domain=0:6.25,smooth,variable=\x] plot ({\x},{6-\x});
      \draw[scale=0.5, domain=0:6.25,smooth,variable=\x, very thick]  plot ({\x},{6/(1-e^(-6)) * (1-e^(-6+\x))});
\end{tikzpicture}
\caption{Normalized risk-averse utility function (bold) and quasilinear utility function for a fixed allocation $\mathbf{x}$ and varying payment $p_i$.}
\label{fig:risk_averse_utility}
\end{SCfigure}
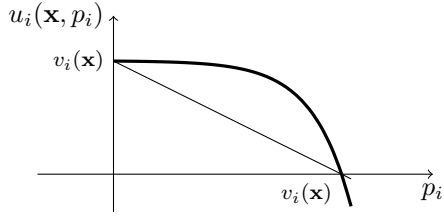

As an effect of normalization, the optimal social welfare of the risk-averse setting can be bounded in terms of the optimal sum of values, which coincides with the social welfare for quasilinear utilities.

\begin{lemma}\label{lemma:OPT_relation}
Given valuation functions $(v_i)_{i \in N}$ and normalized risk-averse utilities $(u_i^{v_i})_{i \in N}$, let $\OPT$ denote the optimal social welfare with respect to utilities $(u_i^{v_i})_{i \in N}$ and $\widehat\OPT$ denote the optimal social welfare with respect to quasilinear utilities $(\hat{u}_i^{v_i})_{i \in N}$. Then, $\OPT \leq 2 \widehat\OPT$.
\end{lemma}

\begin{proof}
Let $\mathbf{x}$, $\mathbf{p}$ denote the outcome and payment profile that maximizes the social welfare $\sum_{i \in N} u_i^{v_i}(\mathbf{x}, p_i) + \sum_{i \in N} p_i$. Consider a fixed player $i$. If $0 \leq p_i \leq v_i(\mathbf{x})$, then by monotonicity of $u_i^{v_i}(\mathbf{x}, \cdot)$ and Assumption~\ref{prop:second}, $u_i^{v_i}(\mathbf{x}, p_i) + p_i \leq u_i^{v_i}(\mathbf{x}, 0) + p_i \leq 2 v_i(\mathbf{x})$. If $p_i > v_i(\mathbf{x})$, then we know from Assumption~\ref{prop:third} that $u_i^{v_i}(\mathbf{x}, p_i) + p_i \leq v_i(\mathbf{x})$. So, always, $u_i^{v_i}(\mathbf{x}, p_i) + p_i \leq 2 v_i(\mathbf{x})$.

By taking the sum over all players, we get $\OPT = \sum_{i \in N} u_i^{v_i}(\mathbf{x}, p_i) + \sum_{i \in N} p_i \leq \sum_{i \in N} 2 v_i(\mathbf{x}) \leq 2 \widehat\OPT$.
\end{proof}

As a consequence, the optimal social welfare changes only within a factor of 2 by risk aversion and we may as well take $\widehat\OPT$ as our point of comparison. A VCG mechanism, for example, is still incentive compatible under risk-averse utilities but optimizes the wrong objective function. Lemma~\ref{lemma:OPT_relation} shows that it is still a constant-factor approximation to optimal social welfare. However, in simple mechanisms, the agents' strategic behavior may or may not change drastically under risk aversion, depending on the mechanism. This way, equilbria and outcomes can possibly be very different.

\section{Smoothness Beyond Quasilinear Utilities}\label{sec:general_smoothness}

Most of our positive results rely on the \emph{smoothness} framework. It was introduced by \cite{Roughgarden09} for general games. There are multiple adaptations to the quasilinear mechanism-design setting. We will use the one by \cite{SyrgkanisT13}. As our utility functions will not be quasilinear, in this section we first \bk{observe that the framework can be extended to general utility functions}. Note that throughout this section, the exact definition of $\OPT(\boldsymbol\theta)$ is irrelevant. Therefore, it can be set to the optimal social welfare but also to weaker benchmarks depending on the setting.

\begin{definition}[Smooth Mechanism]\label{def:smoothness}
A mechanism $M$ is $(\lambda, \mu)$-smooth with respect to utility functions $(u_i^{\theta_i})_{\theta_i\in\Theta_i, i\in N}$ for $\lambda, \mu \ge 0$, if for any type profile $\boldsymbol\theta \in \times_i \Theta_i$ and for any action profile $\mathbf{a}$ there exists a randomized action $a^*_i(\boldsymbol\theta, a_{i})$ for each player $i$, such that
$\sum_i u_i^{\theta_i}(a^*_i(\boldsymbol\theta, a_{i}), \mathbf{a}_{-i}) \ge \lambda \OPT(\boldsymbol\theta) - \mu \sum_i p_i(\mathbf{a})$. We denote by $u_i^{\theta_i}(\mathbf{a})$ the expected utility
of a player if $\mathbf{a}$ is a vector of randomized strategies.
\end{definition}

Mechanism smoothness implies bounds on the price of anarchy. The following theorem and its proof are analogous to the theorems in~\cite{SyrgkanisT13}, the proof is therefore deferred to the Appendix, Section~\ref{sec:missing_proofs}. In cases where the deviation required by smoothness does not depend on $a_i$, the results extend to coarse correlated equilibria. The important point is that the respective bounds mostly do not depend on the assumption of quasilinearity. 

\begin{theorem}\label{thm:smoothness_CCE_and_BNE}
If a mechanism $M$ is $(\lambda, \mu)$-smooth w.r.t.~utility functions $(u_i^{\theta_i})_{\theta_i\in\Theta_i, i\in N}$, then any Correlated Equilibrium in the full information setting and any Bayes-Nash Equilibrium in the Bayesian setting achieves efficiency of at least $\frac{\lambda}{\max\{1, \mu\}}$ of $\OPT(\boldsymbol\theta)$ or of $\mathbb{E}_{\boldsymbol\theta}[\OPT(\boldsymbol\theta)]$, respectively.
\end{theorem}


In the standard single-item setting, one item is auctioned among $n$ players, with their valuations and actions (bids) both being real numbers. In the common auction formats, the item is given to the bidder with the highest bid.

In a \emph{first-price auction}, the winner has to pay her bid; the other players do not pay anything. It is $(1 - 1/e, 1)$-smooth w.r.t.\ quasilinear utility functions. In an \emph{all-pay auction}, all players have to pay their bid. It is $(1/2, 1)$-smooth w.r.t.\ quasilinear utility functions.
These smoothness results were given by \cite{SyrgkanisT13}. They also show that simultaneous and sequential compositions of smooth mechanisms are again smooth.

What is remarkable here is that first-price and all-pay auctions achieve nearly the same welfare guarantees. We will show that in the risk-averse setting this is not true. While the first-price auction almost preserves its constant price of anarchy, the all-pay auction has an unbounded price of anarchy, even with only three players.

\section{Quasilinear Smoothness Often Implies Risk-Averse Smoothness (Main Result~\ref{main-result-smooth})}
\label{sec:positive-results}
Our main positive result is that many price-of-anarchy guarantees that are proved via smoothness in the quasilinear setting transfer to the risk-averse one. First, we consider mechanisms that are $(\lambda, \mu)$-smooth with respect to quasilinear utility functions. We show that if the deviation strategy $\mathbf{a^\ast}$ that is used to establish smoothness ensures non-negative utility, then the price-of-anarchy bound extends to risk-averse settings at a multiplicative constant loss.

%
\begin{theorem}\label{thm:main_result1}
If a mechanism is $(\lambda, \mu)$-smooth w.r.t.~quasilinear utility functions $(\hat{u}_i^{v_i})_{i\in N, v_i\in\mathcal{V}_i}$ and the actions in the support of the smoothness deviations satisfy $\hat{u}_i(a_i^*, \mathbf{a}_{-i})\geq0, \forall \mathbf{a}_{-i}, \forall i$, then any Correlated Equilibrium in the full information setting and any Bayes-Nash Equilibrium in the Bayesian setting achieves efficiency at least $\frac{\lambda}{2\cdot\max\{1, \mu\}}$ of the expected optimal even in the presence of risk averse bidders.
\end{theorem}

Using Theorem~\ref{thm:smoothness_CCE_and_BNE}, it suffices to prove the following lemma.

\begin{lemma}\label{lemma:smooth}
If a mechanism is $(\lambda, \mu)$-smooth w.r.t.~quasilinear utility functions $(\hat{u}_i^{v_i})_{i\in N, v_i\in\mathcal{V}_i}$ and the actions in the support of the smoothness deviations satisfy 
\begin{equation}
\label{eq:nonnegutility}
\hat{u}_i(a_i^*, \mathbf{a}_{-i})\geq0, \forall \mathbf{a}_{-i}, \forall i \enspace,
\end{equation}
then the mechanism is $(\lambda/2, \mu)$-smooth with respect to any normalized risk-averse utility functions $(u_i^{v_i})_{i\in N, v_i\in\mathcal{V}_i}$.
\end{lemma}
\begin{proof}
We start from an arbitrary action profile $\mathbf{a}$ and want to satisfy Definition~\ref{def:smoothness}.
Since there exist smoothness deviations s.t. $\hat{u}_i(a_i^*, \mathbf{a}_{-i})=v_i(\mathbf{x}(a_i^*, \mathbf{a}_{-i}))-p_i\geq0, \forall \mathbf{a}_{-i}, \forall i$, we know from property~\ref{prop:third} of the risk aversion definition that $u_i^{v_i}(a_i^*, \mathbf{a}_{-i})\ge \hat{u}_i^{v_i}(a_i^*, \mathbf{a}_{-i})$. Therefore,
\[
\sum_i u_i^{v_i}(a_i^*, \mathbf{a}_{-i}) \ge \sum_i \hat{u}_i^{v_i}(a_i^*, \mathbf{a}_{-i})\ge \lambda \widehat\OPT - \mu \sum_i p_i(\mathbf{a})\ge\frac{\lambda}{2} \OPT - \mu \sum_i p_i(\mathbf{a})\enspace,
\]
where the last inequality follows from Lemma~\ref{lemma:OPT_relation}.
\end{proof}

Note that in order for~\eqref{eq:nonnegutility} to hold, it is sufficient if all undominated strategies guarantee non-negative quasilinear utility. For example, in a first-price auction, the only undominated bids are the ones from $0$ to $v_i$. Regardless of the other players' bids, these can never result in negative utilities.

\begin{corollary}
Under normalized risk-averse utilities, the first-price auction has a constant price of anarchy for correlated and Bayes-Nash equilibria.
\end{corollary}

\bk{We in addition note here that the first part of Property~\ref{prop:third} of the normalization assumption, $u_i^{v_i}(\mathbf{x}, p_i) \ge v_i(\mathbf{x})-p_i, 0\le p_i\le v_i(\mathbf{x})$, is not crucial for obtaining a result similar to Theorem~\ref{thm:main_result1}. Indeed, a relaxation of the form $u_i^{v_i}(\mathbf{x}, p_i) \ge C\cdot(v_i(\mathbf{x})-p_i), 0\le p_i\le v_i(\mathbf{x})$ for $0<C<1$, $C$ constant, would incur a loss of at most a factor of $C$ in the efficiency bound of Theorem~\ref{thm:main_result1}. More details can be found in the Appendix, Section~\ref{sec:normalization_relaxation}.}

For second-price auctions and their generalizations, for example, the just stated theorems do not suffice to prove guarantees on the quality of equilibria. One in addition needs a no-overbidding assumption and this is further taken care of in the framework of \emph{weak smoothness}, also introduced in~\cite{SyrgkanisT13}. 
We defer all definitions and results that deal with weak smoothness, including the extension from quasilinear to general utility functions and risk-averse utilities yielding a constant loss as compared to the quasilinear case, to the Appendix, Section~\ref{sec:weak_smoothness}.

We also consider the setting where players have hard \emph{budget constraints}. Note that in this case the players' preferences are not quasilinear already in the risk neutral case. Informally, we show that if a mechanism is $(\lambda,\mu)$-smooth w.r.t. quasilinear utility functions, then the loss of performance in the budgeted setting is bounded by a constant, even in the presence of risk-averse agents. All details can be found in Section~\ref{sec:budgets} of the Appendix.


\section{Unbounded Price of Anarchy for All-Pay Auctions (Main Result~\ref{main-result-allpay})}\label{sec:negative-results}
From the previous section, we infer that the constant price-of-anarchy bounds for first-price and second-price auctions immediately extend to the risk-averse setting. This is not true for all-pay auctions; by definition there is no non-trivial bid that always ensures non-negative utility. Indeed, as we show in this section, the price of anarchy is unbounded in the presence of risk-averse players. Missing calculation details can be found in the Appendix, Section~\ref{subsec:calculations}.

\begin{theorem}
\label{theorem:all-pay}
In an all-pay auction with risk-averse players, the PoA is unbounded.
\end{theorem}

The general idea is to construct a Bayes-Nash equilibrium with two players that very rarely have high values and only then bid high values. We then add a third player who always has a high value. However, as the first two players bid high values occasionally, there is no possible bid that ensures he will surely win. This means, any bid has a small probability of not getting the item but having to pay. Risk-averse players are more inclined to avoid this kind of lotteries. In particular, making our third player risk-averse enough, he prefers the sure zero utility of not participating to any way of bidding that always comes with a small probability of negative utility.

\begin{proof}[Proof of Theorem~\ref{theorem:all-pay}]
We consider two (mildly) risk-averse players who both have the same valuation distributions and a third (very) risk-averse player with a constant value. For a large number $M>5$, the first two players have values $v_1$ and $v_2$ drawn independently from distributions with density functions of value $2\cdot(1-(M-1)\cdot\epsilon)$ on the interval $[1/2,1)$ and value $\epsilon$ on the interval $[1,M]$, where $\epsilon=1/M^2$. The third player always has value $1/3\cdot \ln (M/2)$ for winning. 

We will construct a symmetric pure Bayes-Nash equilibrium involving only the first two players. It will be designed such that for the third player it is a best response to always bid $0$, i.e., to opt out of the mechanism and never win the item. So, the combination of these strategies will be a pure Bayes-Nash equilibrium for all three players.

Note that the social welfare of any equilibrium of this form is upper-bounded by the optimal social welfare that can be achieved by the first two bidders. By Lemma~\ref{lemma:OPT_relation}, it is bounded by
\[
\mathbb{E}[\SW]\le2\cdot \mathbb{E}[\max\{v_1,v_2\}]\le2\cdot \mathbb{E}[v_1+v_2]=2\cdot (\mathbb{E}[v_1]+\mathbb{E}[v_2])=4\cdot\mathbb{E}[v_1]\le 4 \enspace.
\]
Furthermore, the third player's value $v_3 = 1/3\cdot \ln (M/2)$ is a lower bound to the optimal social welfare \bk{in the construction containing all three players}. So, as pointwise $\OPT(v)\ge 1/3\cdot\ln(M/2)$, where $v=(v_1,v_2,v_3)\in\mathcal{V}$ denotes the valuation profile, this implies that the price of anarchy can be arbitrarily high.

We define the utility functions by setting
\begin{equation}
u^{v_i}_i(b_i) =
\begin{cases}
\frac{h_i(v_i-b_i)}{h(v_i)}\cdot v_i & \text{if $b_i$ is the winning bid}\\
\frac{h_i(-b_i)}{h(v_i)}\cdot v_i & \text{otherwise } 
\end{cases}
\end{equation}
For the first two players, we use $h_i(x):=1-e^{-x}$, $i \in \{1, 2\}$, so in particular increasing and concave. For the third player, we set $h_i(x) = x$ for $x \geq 0$ and $h_i(x) = C \cdot x$ for $x < 0$, where $C = (16\cdot\frac{1}{3}\cdot\ln M/2)\cdot M^2 \geq 1$. Again this function is increasing and concave\footnote{Its slope is not an absolute constant. This is indeed necessary because the price of anarchy can be bounded in terms of the slopes of the $h_i$-functions as we show in Appendix~\ref{sec:bounded_slope}.}. \bk{Note that the utility functions also satisfy normalizations at $p_i=v_i(\mathbf{x})$ and at $p_i=0$.} We see that in this example risk aversion has the effect of heavily penalizing payments without winning the auction.

\begin{claim}
With the third player not participating, it is a symmetric pure Bayes-Nash equilibrium for the first two players to play according to bidding function $\beta\colon\mathcal{V}_i\to \mathbb{R}_+, i\in\{1,2\},$ such that 
\begin{equation}\label{eq:beta}
\beta(x) =\int_{\frac{1}{2}}^x \frac{f(t)(e^t-1)}{F(t)+(1-F(t))e^t} dt\enspace,
\end{equation}
where $F$ denotes the cumulative distribution function of the value and $f$ denotes its density.
\end{claim}
\begin{proof}
We will argue that playing according to $\beta$ is always the unique best response if the other player is playing according to $\beta$, too. Due to symmetry reasons, it is enough to argue about the first player.

Let us fix player 1's value $v_1 = x$ and consider the function $g\colon \mathbb{R}_{\geq 0} \to \mathbb{R}$ that is defined by $g(y) = \mathbf{E}[u_1^x(b_1=y, b_2=\beta(v_2), b_3=0)]$. We claim that $g$ is indeed maximized at $y = \beta(x)$. We have\footnote{Note that the first step assumes tie breaking in favor of player 1. This is irrelevant for the future steps as the involved probability distributions are continuous.}
\begin{align*}
g(y)&=\Pr[\beta(v_2) \leq y]\cdot \frac{h_1(x-y)}{h_1(x)}\cdot x + \left(1-\Pr[\beta(v_2) \leq y]\right)\cdot \frac{h_1(-y)}{h_1(x)}\cdot x\\
&=\frac{x}{h_1(x)}\left[
F(\beta^{-1}(y))\Big(h_1(x-y)-h_1(-y)\Big) + h_1(-y)\right] \\
&=x e^y F(\beta^{-1}(y)) + \frac{x ( 1 - e^y )}{1 - e^{-x}}\enspace.
\end{align*}
The first derivative of this function is given by
\begin{align*}
g'(y) & = x e^y F(\beta^{-1}(y)) + x e^y \frac{d}{dy} F(\beta^{-1}(y)) - \frac{x}{1 - e^{-x}} e^y\enspace.
\end{align*}
The inverse function theorem implies $\frac{d}{dy} F(\beta^{-1}(y)) = \frac{f(\beta^{-1}(y))}{\beta'(\beta^{-1}(y))}$.
Furthermore, as $\beta'(t) = \frac{f(t)(e^t-1)}{F(t)+(1-F(t))e^t}$, we get for all $t$ that 
$
\frac{f(t)}{\beta'(t)} = \frac{F(t) + (1 - F(t)) e^t}{e^t - 1}= (1-F(t)) + \frac{1}{e^t - 1}\enspace.
$
This simplifies $g'(y)$ to
\begin{align*}
g'(y) & =xe^y+ \frac{x e^y}{e^{\beta^{-1}(y)} - 1} - \frac{xe^y}{1 - e^{-x}} =\frac{xe^y}{(1-e^{-x})(e^{\beta^{-1}(y)}-1)} \left(1 - e^{-x+\beta^{-1}(y)} \right) \enspace.
\end{align*}
Notice that the factor $\frac{xe^y}{(1-e^{-x})(e^{\beta^{-1}(y)}-1)}$ is always positive. Therefore, we observe that $g'(y) = 0$ if and only if $e^{-x+\beta^{-1}(y)} = 1$, which is equivalent to $y = \beta(x)$. Furthermore, $g'(y) > 0$ for $y < \beta(x)$ and $g'(y) < 0$ for $y > \beta(x)$. This means that $y = \beta(x)$ has to be the (unique) global maximum of $g(y)$.
 \end{proof}

\begin{claim}
If the first two players are bidding according to~(\ref{eq:beta}), then it is a best response for the third player to always bid $0$.
\end{claim}
\begin{proof}
We now show that the very risk-averse third player with valuation $1/3\cdot \ln (M/2)$ will indeed bid $0$ because every bid $b_3' > 0$ will cause negative expected utility.

We distinguish two cases. For values of $b_3' > \frac{1}{16}$, we use that with a small probability one of the two remaining players has a valuation of at least $M-1$, which leads to negative utility. For $b_3' \leq \frac{1}{16}$ on the other hand, he loses so often that his expected utility is again negative. 

Let us first assume that the third player bids $b_3'$ with $\frac{1}{16} < b_3' \leq v_3$. In this case, with probability more than $\epsilon$ one of the first two players has value of at least $M-1$. The bid of this player with $v_i \geq M-1$ can be estimated as follows
\begin{align*}
\beta(v_i)&\geq \beta(M-1) \geq \int_{M/2}^{M-1} \frac{f(t)(e^t-1)}{1+(1-F(t))e^t}dt = \int_{M/2}^{M-1}\frac{\epsilon (e^t-1)}{1+\epsilon(M-t)e^t}dt\\ &\ge \frac{1}{2}(1-e^{-\frac{M}{2}}) \int_{M/2}^{M-1} \frac{\epsilon e^t}{\epsilon (M-t)e^t}dt = \frac{1}{2}(1-e^{-\frac{M}{2}}) \ln (M/2) > \frac{1}{3} \ln (M/2)\enspace,
\end{align*}
which means that by bidding $b_3'$ the third player loses with probability of at least $\epsilon=1/M^2$. For the expected utility, this implies
\begin{align*}
\mathbb{E}[u_3(b_3', \mathbf{b}_{-3})] &\leq (1 - \epsilon)(v_3 - b_3') + \epsilon(- C\cdot b_3')
< \frac{1}{3}\ln \frac{M}{2} - \frac{1}{16}\cdot 16\cdot\frac{1}{3}\ln \frac{M}{2} = 0\enspace.
\end{align*}

In the case where the third player bids $b_3'$, $0 < b_3' \le \frac{1}{16}$, we need to be a bit more careful with estimating the winning probability. We first give a lower bound on the bidding function of the first player \bk{for $v_1<1$} 
\begin{align*}\label{eq:bid_estimate}
\beta(v_1)&\ge \int_{1/2}^{v_1} \frac{2(1-\frac{M-1}{M^2})(e^t-1)}{2(t-\frac{1}{2})(1-\frac{M-1}{M^2})+ (1-2(t-\frac{1}{2})(1-\frac{M-1}{M^2}))\cdot e^t} dt>\frac{1}{4} \left( v_1 - \frac{1}{2} \right)\enspace.
\end{align*}
\bk{Since for $v_1\ge1$, $\beta(v_1)>\frac{1}{16}$ with probability $1$,} this implies that with $b_3'$, the third player has a winning probability of at most
\[
\Pr[\beta(v_1) \leq b_3'] \le \Pr\left[\frac{1}{4}\left(v_1-\frac{1}{2}\right) \leq b_3'\right] = \Pr\left[v_1 \leq 4b_3'+\frac{1}{2}\right]<2\cdot 4 b_3'\enspace.
\]
Now, having in mind that $C = (16v_3)\cdot M^2\ge 32\cdot v_3$, the utility can be estimated as follows
\begin{align*}
\mathbb{E}[u_3(b_3', \mathbf{b}_{-3})] &\le\Pr[\beta(v_1) \leq b_3']\cdot(v_3-b_3')-\Pr[\beta(v_1) > b_3']\cdot 32\cdot v_3\cdot b_3' \\
&<8b_3'\left(-v_3-\frac{1}{16}\right)<0\enspace.
\end{align*}
So also in this case, the expected utility is negative.
\end{proof}
Combining the two claims, we have constructed a class of distributions and Bayes-Nash equilibria with unbounded price of anarchy.
\end{proof}

As a final remark, we note that the first two bidders occasionally bid high only due to risk aversion. In a symmetric Bayes-Nash equilibrium of the all-pay auction in the quasilinear setting, all bids are always bounded by the expected value of a player (see Appendix~\ref{sec:quasilinear_allpay}). Therefore, such an equilibrium would not work as a point of departure.

\section{Variance-Aversion Model}\label{sec:other-models}
In this section, we consider a different model that tries to capture the effect that agents prefer certain outcomes to uncertain ones. It is inspired by similar models in game theory and penalizes variance of random variables. Rather than reflecting the aversion in the utility functions, it is modeled by adapting the solution concept. 

In the usual definition of equilibria involving randomization, the utility of a randomized strategy profile is set to be the expectation over the pure strategies. The definition we consider here is modified by subtracting the respective standard deviation. For a player $i$, the utility of a randomized strategy profile $\mathbf{a}$ is given as
$u^{v_i}_i(\mathbf{a}) = \mathbb{E}_{\mathbf{b} \sim \mathbf{a}}[\hat{u}^{v_i}_i(\mathbf{b})] - \gamma\sqrt{\Var[\hat{u}_i^{v_i}(\mathbf{b})]}$, 
so a player's utility for an action profile is his expected quasilinear utility for this profile minus the standard deviation multiplied by a parameter $\gamma$ that determines the degree of variance-averseness, $0\le\gamma\le1$. As already mentioned, $\hat{u}_i(\mathbf{a})$ denotes the quasilinear utility of player $i$ for the action profile $\mathbf{a}$.

Bayes-Nash Equilibria and correlated equilibria can be defined the same way as before, always replacing expectations by the difference of expectation and standard deviation. The formal definition for $s(\mathbf{v})$ being a Bayes-Nash equilibrium in this setting is that $\forall i \in N$, $\forall v_i\in \Theta_i$, $a_i\in \mathcal{A}_i$, 
\begin{multline*}
\mathbb{E}_{\mathbf{v}_{-i}}[\hat{u}_i^{v_i}(s(\mathbf{v})) \mid v_i] - \gamma\sqrt{\Var[\hat{u}_i^{v_i}(s(\mathbf{v})) \mid v_i]}\\
\ge \mathbb{E}_{\mathbf{v}_{-i}}[\hat{u}_i^{v_i}(a_i,s_{-i}(\mathbf{v}_{-i})) \mid v_i] - \gamma\sqrt{\Var[\hat{u}_i^{v_i}(a_i, s_{-i}(\mathbf{v}_{-i})) \mid v_i]}\enspace.
\end{multline*}

Note that we again evaluate social welfare as agents perceive it. That is, for a randomized strategy profile $\mathbf{a}$, we set $\SW^{\mathbf{v}}(\mathbf{a}) = \sum_i u_i^{v_i}(\mathbf{a}) + \sum_i p_i(\mathbf{a})$. 

Our first result shows that first-price and notably also all-pay auctions have a constant price of anarchy in this setting. Note that even though the proof looks a lot like smoothness proofs, it is not possible to phrase it within the smoothness framework, since here we are dealing with a different solution concept.

\begin{theorem}
Bayes-Nash Equilibria and Correlated Equilibria of the first-price and all-pay auction have a constant price of anarchy in this model.
\end{theorem}
\begin{proof}
For simplicity, we will show the claim only for Bayes-Nash equilibria. The proof for correlated equilibria works the same way with minor modifications to the notation.

Assume $\mathbf{b}$ is a Bayes-Nash equilibrium. We claim that
$
\mathbb{E}_{\mathbf{v}}\left[\SW^{\mathbf{v}}(\mathbf{b})\right]\ge \frac{1}{16} \cdot \mathbb{E}_{\mathbf{v}}[\OPT]\enspace,
$
where $\OPT$ denotes the value of social welfare in the allocation that maximizes it, i.e. maximized sum of utility and payments of the agents.

Consider a fixed player $j$ and a fixed valuation $v_j$. Let $q=\Pr[\max_{i\neq j}b_{i} \le \frac{1}{4}\cdot v_j]$ denote the probability that no other player's bid exceeds $\frac{1}{4}\cdot v_j$.

Assume first that $q\le \frac{3}{4}$. Then, because the total social welfare lower bounded by the payments
$
\mathbb{E}_{\mathbf{v}_{-j}\vert v_j}\left[\SW^{\mathbf{v}}(\mathbf{b})\right]
\ge (1-q)\frac{1}{4} v_j\ge\frac{1}{16} v_j\enspace.
$

On the other hand, if $q\ge\frac{3}{4}$, since $\mathbb{E}_{\mathbf{v}_{-j}\vert v_j}\left[\SW^{\mathbf{v}}(\mathbf{b})\right]\ge \mathbb{E}_{\mathbf{v}_{-j}\vert v_j}\left[u_j^{v_j}(\frac{v_j}{4},b_{-j})\right]$,
\begin{multline*}\label{eq:variance_first_and_allpay}
\mathbb{E}_{\mathbf{v}_{-j}\vert v_j}\left[\SW^{\mathbf{v}}(\mathbf{b})\right]
 \ge v_j q - \frac{1}{4} v_j - \gamma v_j\sqrt{q(1-q)} \ge \left(\frac{2- \gamma\sqrt{3}}{4}\right)v_j\ge \frac{1}{16}v_j\enspace,
\end{multline*}
where the first inequality is in fact an equality for the all-pay auction.

From here, by taking the expectation over $v_j$ and by weighing the right hand side by the probability that $\OPT$ takes a particular agent, the theorem follows.
\end{proof}

This is contrasted by a correlated equilibrium with $0$ social welfare in a setting with positive values. Indeed, for the special case of $\lambda=1$, we see that the variance-averse model further differs from the risk-averse model described in previous sections.  

\begin{observation}\label{lemma:observation}
The PoA for CE of second price auctions is unbounded if $\gamma=1$. 
\end{observation}

The proof can be found in the Appendix, Section~\ref{sec:observation}. This is not only a difference between smoothness and weak smoothness. Our final result is a mechanism that is $(\lambda,\mu)$-smooth for constant $\lambda$ and $\mu$ but has unbounded price of anarchy. 

\begin{theorem}\label{thm:variance_smoothness}
For any constant $\gamma>0$ there is a mechanism that is $(\lambda,\mu)$-smooth with respect to quasilinear utility functions for constant $\lambda$ and $\mu$ but has unbounded price of anarchy in the variance-aversion model.
\end{theorem}
\begin{proof}
Consider a setting with two items and two players, who have unit-demand valuation functions such that $\frac{1}{c} v_{i, 1} \leq v_{i, 2} \leq c v_{i, 1}$ for constant $c \geq 1$. The players' possible actions are to either report one of the two items as preferred or to opt out entirely. Our mechanism first assigns player 1 her (claimed) favorite item, then assigns player 2 the remaining one unless she opts out. There are no payments. Obviously, this mechanism is $(\frac{1}{c}, 0)$-smooth because the allocation is within a $\frac{1}{c}$-factor of the optimal allocation by construction.

We will now construct a mixed Nash equilibrium of bad welfare. To this end, let $v_{1, 1} = v_{1, 2} = \epsilon$ for some small $\epsilon > 0$. This makes player 1 indifferent between items 1 and 2. In particular, it is a best response to ask for item 1 with probability $\frac{q-1}{q}$ and for item 2 with probability $\frac{1}{q}$. We note at this point that in a Bayes-Nash equilibrium we could make this respective action the unique best response by having random types.

For player 2, we set $v_{2, 1} = c$, $v_{2, 2} = 1$. She has the choice of participating or opting out. Opting out implies utility $0$, whereas participating implies
\begin{align*}
u_2(\mathbf{a}) 
&= \frac{c+q-1}{q} - \gamma\sqrt{\frac{(c-1)^2(q-1)}{q^2}}= \frac{(c-1)(1-\gamma\sqrt{q-1})}{q}+1
\end{align*}
Now, if we set $q=c-1$, then $u_2(\mathbf{a})=2-\gamma\sqrt{c-2}$ which is negative for $c>\frac{4}{\gamma^2}+2$. We further set $c=\frac{4}{\gamma^2}+3$.
That is, player 2 prefers to opt out. This outcome has social welfare $\epsilon$ whereas the optimal social welfare is $c$.
\end{proof}

Note that this last example shows that variance-averseness yields very strange preferences for lotteries. In our example, the variance-averse player prefers not to participate although any outcome in the (free) lottery has positive value.

\bibliography{bibliography}


\newpage
\appendix
\section{Proof of Theorem~\ref{thm:smoothness_CCE_and_BNE}}\label{sec:missing_proofs}
\subsection{Full Information Setting}
\begin{proof}
Let $\bf a$ be a correlated equilibrium. This means that for every $a_i$ in the support of $\bf a$
\[
\mathbb{E}_{\mathbf{a}_{-i}\vert a_i}[u_i^{\theta_i}(a_i, \mathbf{a}_{-i})]\ge \mathbb{E}_{\mathbf{a}_{-i}\vert a_i}[u_i^{\theta_i}(a_i', \mathbf{a}_{-i})], \forall a_i'\in\mathcal{A}_i, \forall i\enspace.
\]
Applying the equilibrium property to $a_i'=a_i^*(\boldsymbol{\theta}, a_i)$, we know that for every $a_i$ in the support of $\bf a$:
\[
\mathbb{E}_{\mathbf{a}_{-i}\vert a_i}[u_i^{\theta_i}(a_i, \mathbf{a}_{-i})]\ge \mathbb{E}_{\mathbf{a}_{-i}\vert a_i}[u_i^{\theta_i}(a_i^*(\boldsymbol\theta, a_i), \mathbf{a}_{-i})], \forall i\enspace.
\]
If we now take the expectation over $a_i$ and add over all players:
\[
\mathbb{E}_{\mathbf{a}}[\sum_i u_i^{\theta_i}(\mathbf{a})]\ge \mathbb{E}_{\mathbf{a}}[\sum_i u_i^{\theta_i}(a_i^*(\boldsymbol\theta, a_i), \mathbf{a}_{-i})]\ge \lambda\OPT(\boldsymbol\theta) - \mu\mathbb{E}_{\mathbf{a}}[\sum_i p_i(\mathbf{a})]\enspace,
\]
and further by adding  $\mathbb{E}_{\mathbf{a}}[\sum_i p_i(\mathbf{a})]$ to both sides
\begin{multline*}
\mathbb{E}_{\mathbf{a}}[\sum_i u_i^{\theta_i}(\mathbf{a}) + \sum_i p_i(\mathbf{a})]\ge \mathbb{E}_{\mathbf{a}}[\sum_i u_i^{\theta_i}(a_i^*(\boldsymbol\theta, a_i), \mathbf{a}_{-i})]\\
\ge \lambda\OPT(\boldsymbol\theta) + (1 - \mu)\mathbb{E}_{\mathbf{a}}[\sum_i p_i(\mathbf{a})]\enspace.
\end{multline*}
The result follows by doing a case distinction over $\mu\le1$ and $\mu>1$. In the first case, we immediately get
\[
\mathbb{E}_{\mathbf{a}}[\sum_i u_i^{\theta_i}(\mathbf{a}) + \sum_i p_i(\mathbf{a})]
\ge \lambda\OPT(\boldsymbol\theta) + (1 - \mu)\mathbb{E}_{\mathbf{a}}[\sum_i p_i(\mathbf{a})]
\ge \lambda\OPT(\boldsymbol\theta)\enspace,
\]
and in the latter case we use the fact that $\mathbb{E}_{\mathbf{a}}[\sum_i u_i^{\theta_i}(\mathbf{a})]\ge 0$. Then also 
\[
\mathbb{E}_{\mathbf{a}}[\sum_i u_i^{\theta_i}(\mathbf{a})] + \mathbb{E}_{\mathbf{a}}[\sum_i p_i(\mathbf{a})]\ge \mathbb{E}_{\mathbf{a}}[\sum_i p_i(\mathbf{a})]\enspace,
\] 
which results in 
\[
\mathbb{E}_{\mathbf{a}}[\sum_i u_i^{\theta_i}(\mathbf{a}) + \sum_i p_i(\mathbf{a})]
\ge \lambda\OPT(\boldsymbol\theta) + (1 - \mu)\mathbb{E}_{\mathbf{a}}[\sum_i u_i^{\theta_i}(\mathbf{a}) + \sum_i p_i(\mathbf{a})]
\] 
and finally 
\[\mathbb{E}_{\mathbf{a}}[\sum_i u_i^{\theta_i}(\mathbf{a}) + \sum_i p_i(\mathbf{a})]
\ge \frac{\lambda}{\mu}\OPT(\boldsymbol\theta)\enspace.
\]
\end{proof}

\subsection{Bayesian Setting}
\begin{proof}
For reasons of clarity, we prove the claim for the simpler case of pure BNE. First, we let each player $i$ sample a type profile $\boldsymbol\zeta\sim \times_i F_i$ and play $a_i^*((\theta_i, \boldsymbol\zeta_{-i}), s_i(\zeta_i))$.

\begin{align*}
\mathbb{E}_{\boldsymbol\theta}[u_i^{\theta_i}(s(\boldsymbol\theta))]&\geq \mathbb{E}_{\boldsymbol\theta,\boldsymbol\zeta}[u_i^{\theta_i}(a_i^*((\theta_i, \boldsymbol\zeta_{-i}), s_i(\zeta_i)), s_{-i}(\boldsymbol\theta_{-i}))]\\
&=\mathbb{E}_{\boldsymbol\theta,\boldsymbol\zeta}[u_i^{\zeta_i}(a_i^*((\zeta_i, \boldsymbol\zeta_{-i}), s_i(\theta_i)), s_{-i}(\boldsymbol\theta_{-i}))]\\
&=\mathbb{E}_{\boldsymbol\theta,\boldsymbol\zeta}[u_i^{\zeta_i}(a_i^*(\boldsymbol\zeta, s_i(\theta_i)), s_{-i}(\boldsymbol\theta_{-i}))]
\end{align*}

Summing over the players and using the smoothness property, we get
\begin{multline*}
\mathbb{E}_{\boldsymbol\theta}\Bigg[\sum_i u_i^{\theta_i}(s(\boldsymbol\theta))\Bigg]
\geq \mathbb{E}_{\boldsymbol\theta, \boldsymbol\zeta}\Bigg[\sum_i u_i^{\zeta_i}(a_i^*(\boldsymbol\zeta, s_i(\theta_i)), s_{-i}(\boldsymbol\theta_{-i}))\Bigg]\\\geq\mathbb{E}_{\boldsymbol\theta, \boldsymbol\zeta}\Bigg[\lambda\OPT(\boldsymbol\zeta)-\mu\sum_i p_i(s(\boldsymbol\theta))\Bigg]=\lambda\mathbb{E}_{\boldsymbol\theta}[\OPT(\boldsymbol\theta)] - \mu\mathbb{E}_{\boldsymbol\theta}\Bigg[\sum_i p_i(s(\boldsymbol\theta))\Bigg]\enspace,
\end{multline*}
and therefore
\[
\mathbb{E}_{\boldsymbol\theta}\Bigg[\sum_i u_i^{\theta_i}(s(\boldsymbol\theta)) + \sum_i p_i(s(\boldsymbol\theta))\Bigg]\geq \lambda\mathbb{E}_{\boldsymbol\theta}[\OPT(\boldsymbol\theta)] + (1-\mu)\mathbb{E}_{\boldsymbol\theta}\Bigg[\sum_i p_i(s(\boldsymbol\theta))\Bigg]\enspace,
\]
from where the result follows by case distinction over $\mu$, as in the proof for the full information setting.

The generalization to a mixed Bayes-Nash equilibrium is now straightforward.
\end{proof}

\section{Relaxation of the normalization assumption}\label{sec:normalization_relaxation}

We assume the following \emph{relaxed normalized risk-averse utilities}:
\begin{enumerate}
\item \label{prop:monotonicity_r} $u_i^{v_i}(\mathbf{x}, p_i) \geq u_i^{v_i}(\mathbf{x}, p_i')$ if $p_i \leq p_i'$ (monotonicity)
\item \label{prop:first_r} $u_i^{v_i}(\mathbf{x}, p_i)=0$ if $p_i=v_i(\mathbf{x})$ (normalization at $p_i=v_i(\mathbf{x})$)
\item \label{prop:second_r} $u_i^{v_i}(\mathbf{x},p_i)=v_i(\mathbf{x})$ if $p_i=0$ (normalization at $p_i=0$)
\item \label{prop:third_r} $u_i^{v_i}(\mathbf{x},p_i)\geq C\cdot(v_i(\mathbf{x})-p_i)$ if $0\leq p_i\leq v_i(\mathbf{x}), 0<C<1, C$ constant and $u_i^{v_i}(\mathbf{x},p_i)\leq v_i(\mathbf{x})-p_i$ otherwise (extra relaxed concavity)
\end{enumerate}

\begin{lemma}\label{lemma:smooth}
If a mechanism is $(\lambda, \mu)$-smooth w.r.t. quasilinear utility functions $(\hat{u}_i^{v_i})_{i\in N, v_i\in\mathcal{V}_i}$ and the actions in the support of the smoothness deviations satisfy 
\begin{equation}
\hat{u}_i(a_i^*, \mathbf{a}_{-i})\geq0, \forall \mathbf{a}_{-i}, \forall i \enspace,
\end{equation}
then the mechanism is $(C\cdot\lambda/2, C\cdot\mu)$-smooth with respect to any relaxed normalized risk-averse utility functions $(u_i^{v_i})_{i\in N, v_i\in\mathcal{V}_i}$.
\end{lemma}
\begin{proof}
We start from an arbitrary action profile $\mathbf{a}$ and want to satisfy Definition~\ref{def:smoothness}.
Since there exist smoothness deviations s.t. $\hat{u}_i(a_i^*, \mathbf{a}_{-i})=v_i(\mathbf{x}(a_i^*, \mathbf{a}_{-i}))-p_i\geq0, \forall \mathbf{a}_{-i}, \forall i$, we know from property~\ref{prop:third_r} of the relaxed risk aversion definition that $u_i^{v_i}(a_i^*, \mathbf{a}_{-i})\ge C\cdot\hat{u}_i^{v_i}(a_i^*, \mathbf{a}_{-i})$. Therefore,
\begin{align*}
\sum_i u_i^{v_i}(a_i^*, \mathbf{a}_{-i}) &\ge \sum_i C\cdot\hat{u}_i^{v_i}(a_i^*, \mathbf{a}_{-i})\ge C\cdot\lambda \widehat\OPT - C\cdot\mu \sum_i p_i(\mathbf{a})\\
&\ge\frac{C\cdot\lambda}{2} \OPT - C\cdot\mu \sum_i p_i(\mathbf{a})\enspace,
\end{align*}
where the last inequality follows from Lemma~\ref{lemma:OPT_relation}.
\end{proof}

Using Theorem~\ref{thm:smoothness_CCE_and_BNE}, we obtain the following theorem.

\begin{theorem}
If a mechanism is $(\lambda, \mu)$-smooth w.r.t. quasilinear utility functions $(\hat{u}_i^{v_i})_{i\in N, v_i\in\mathcal{V}_i}$ and the actions in the support of the smoothness deviations satisfy $\hat{u}_i(a_i^*, \mathbf{a}_{-i})\geq0, \forall \mathbf{a}_{-i}, \forall i$, then any Correlated Equilibrium in the full information setting and any Bayes-Nash Equilibrium in the Bayesian setting achieves efficiency at least $\frac{C\cdot\lambda}{2\cdot\max\{1, C\cdot\mu\}}$ of the expected optimal even in the presence of risk averse bidders.
\end{theorem}

\section{Weak Smoothness}\label{sec:weak_smoothness}
\subsection{Extension to General Utility Functions}

For second-price auctions and their generalizations, for example, the already stated theorems do not suffice to prove guarantees on the quality of equilibria. One in addition needs a no-overbidding assumption. To state this assumption, we first need the notion of willingness-to-pay that was originally defined in \cite{SyrgkanisT13}.


%

\begin{definition}[Willingness-to-pay]
Given a mechanism $(\mathcal{A}, X, P)$ a player's maximum willingness-to-pay for an allocation $\mathbf{x}$ when using strategy $a_i$ is defined as the maximum he could ever pay conditional on allocation $\mathbf{x}$:
\[
W_i(a_i,\mathbf{x})=\max_{\mathbf{a}_{-i}:X(\mathbf{a})=\mathbf{x}} p_i(\mathbf{a})\enspace.
\]
\end{definition}

Now, we can state weak smoothness.

\begin{definition}[Weakly Smooth Mechanism]\label{def:weaksmoothness}
A mechanism $M$ is weakly $(\lambda, \mu_1, \mu_2)$-smooth with respect to utility functions $(u_i^{\theta_i})_{\theta_i\in\Theta_i, i\in N}$ for $\lambda, \mu_1, \mu_2 \ge 0$, if for any type profile $\boldsymbol\theta \in \times_i \Theta_i$ and for any action profile $\mathbf{a}$ there exists a randomized action $a^*_i(\boldsymbol\theta, a_{i})$ for each player $i$, s.t.:
\begin{equation*}\label{eq:weaksmoothness}
\sum_i u_i^{\theta_i}(a^*_i(\boldsymbol\theta, a_{i}), \mathbf{a}_{-i}) \ge \lambda \OPT(\boldsymbol\theta) - \mu_1 \sum_i p_i(\mathbf{a}) - \mu_2 \sum_i W_i(a_i, X(\mathbf{a}))\enspace. 
\end{equation*}
We denote by $u_i^{\theta_i}(\mathbf{a})$ the expected utility
of a player if $\mathbf{a}$ is a vector of randomized strategies.
\end{definition} 

Note that $(\lambda, \mu)$-smoothness implies weak $(\lambda, \mu, 0)$-smoothness. We get the following generalization of the price-of-anarchy guarantees for equilibria that fulfill the aforementioned no-overbidding assumption on the players' willigness-to-pay:

\begin{theorem}\label{thm:weak_smoothness}
If a mechanism is weakly $(\lambda, \mu_1, \mu_2)$-smooth w.r.t. utility functions $(u_i^{\theta_i})_{\theta_i\in\Theta_i, i\in N}$, then any Correlated Equilibrium in the full information setting and any Bayes-Nash Equilibrium in the Bayesian setting that satisfies
\begin{equation}
\mathbb{E}_{\mathbf{a}}[W_i( a_i, X( \mathbf{a}))] \leq \mathbb{E}_{\mathbf{a}}[u_i^{\theta_i}( \mathbf{a}) + p_i( \mathbf{a})]
\label{eq:generic_no-overbidding}
\end{equation}
achieves efficiency of at least $\frac{\lambda}{(\mu_2 + \max\{\mu_1,1\})}$ of $\OPT(\boldsymbol\theta)$ or of $\mathbb{E}_{\boldsymbol\theta}[\OPT(\boldsymbol\theta)]$, respectively.
\end{theorem}

In the quasilinear setting, \eqref{eq:generic_no-overbidding} simplifies to the no-overbidding assumption $\mathbb{E}_{\mathbf{a}}[W_i(a_i, X(\mathbf{a}))]\le \mathbb{E}_{\mathbf{a}}[v_i(X(\mathbf{a}))]$ that was introduced in \cite{SyrgkanisT13}, and that is a generalization of the no-overbidding assumptions previously used in the literature~\cite{ChristodoulouKS08,BhawalkarR11,CaragiannisKKKLLT15}. That is, players cannot pay more than their respective value, regardless of the other players' actions.

\begin{proof}
For the complete information setting, we show that
\[
\sum_i u_i^{\theta_i} (\mathbf{a})\geq \lambda\cdot \widehat{\OPT} - \mu_1\sum_i p_i(\mathbf{a}) - \mu_2 \sum_i W_i(a_i, X_i(\mathbf{a}))\enspace.
\]
Using~\eqref{eq:generic_no-overbidding},
\[
(1+\mu_2)\Bigg[\sum_i u_i^{\theta_i}(\mathbf{a}) + \sum_i p_i(\mathbf{a})\Bigg]\ge \lambda \widehat{\OPT} - (\mu_1-1)\sum_i p_i(\mathbf{a})\enspace.
\]
By doing a case distinction over $\mu\le1$ and $\mu>1$, we get the claimed result.

For the incomplete information setting, we arrive at
\begin{multline*}
\mathbb{E}_{\boldsymbol\theta}\Bigg[\sum_i u_i^{\theta_i}(s(\boldsymbol\theta))\Bigg]\geq \lambda\mathbb{E}_{\boldsymbol\theta}[\widehat{\OPT}] - \mu_1\mathbb{E}_{\boldsymbol\theta}\Bigg[\sum_i p_i(s(\boldsymbol\theta))\Bigg]\\
-\mu_2 \mathbb{E}_{\boldsymbol\theta}\Bigg[\sum_i W_i(s_i(\theta_i), X_i(s(\boldsymbol\theta)))\Bigg]\enspace.
\end{multline*}
The result now follows by using the no-overbidding assumption~\eqref{eq:generic_no-overbidding} and case distinction, similar to the full information case.
\end{proof}

In a \emph{second-price auction}, the winner has to pay the second highest bid, the other players do not pay anything. In the quasilinear setting it is weakly $(1, 0, 1)$-smooth.

\subsection{Risk-Averse Utilities}
We will assume the following pointwise condition:

\begin{definition}[Pointwise No-Overbidding]
A randomized strategy profile $\mathbf{a}$ satisfies the pointwise no-overbidding assumption if for every player $i$ and every action in the support of $\mathbf{a}$ the following holds:
\[
W_i(a_i, \mathbf{x}):= \max_{\mathbf{a}_{-i}:X(\mathbf{a})=\mathbf{x}} p_i(\mathbf{a}) \le v_i(\mathbf{x})\enspace,
\]
i.e. no player is pointwise bidding in a way that she could potentially pay more than her value, subject to her allocation remaining the same.
\end{definition}

\begin{theorem}\label{thm:weak_smoothness_risk}
If a mechanism is weakly $(\lambda, \mu_1, \mu_2)$-smooth with respect to quasilinar utility functions $(\hat{u}_i^{v_i})_{i\in N, v_i\in\mathcal{V}_i}$, the actions in the support of the smoothness deviations satisfy $\hat{u}_i(a_i^*, \mathbf{a}_{-i})\geq0, \forall \mathbf{a}_{-i}, \forall i$, then any Correlated Equilibrium in the full information setting and any Bayes-Nash Equilibrium in the Bayesian setting that satisfies the pointwise no-overbidding assumption achieves efficiency at least $\frac{\lambda}{2\cdot(\mu_2 + \max\{\mu_1,1\})}$ of the expected optimal even in the presence of risk-averse bidders.
\end{theorem}
\begin{proof}
First, we show that weak smoothness with respect to quasilinear utility functions with the additional constraint that players have non-negative utility from the smoothness deviation implies weak smoothness with respect to risk-averse players. 

\begin{lemma}\label{lemma:weakly_smooth}
If a mechanism is weakly $(\lambda, \mu_1, \mu_2)$-smooth with respect to quasilinear utility functions $(\hat{u}_i^{v_i})_{i\in N, v_i\in\mathcal{V}_i}$ and the actions in the support of the smoothness deviations satisfy $\hat{u}_i(a_i^*, \mathbf{a}_{-i})\geq0, \forall a_{-i}, \forall i$, then the mechanism is weakly $(\lambda/2, \mu_1, \mu_2)$-smooth with respect to risk-averse utility functions $(u_i^{v_i})_{i\in N, v_i\in\mathcal{V}_i}$.
\end{lemma}
\begin{proof}
We start from an arbitrary action profile $\mathbf{a}$ and want to satisfy Definition~\ref{def:smoothness}.
Since there exist smoothness deviations s.t. $\hat{u}_i(a_i^*, \mathbf{a}_{-i})=v_i(\mathbf{x}(a_i^*, \mathbf{a}_{-i}))-p_i\geq0, \forall \mathbf{a}_{-i}, \forall i$, we know from property~\ref{prop:third} of the risk aversion definition that $u_i^{v_i}(a_i^*, \mathbf{a}_{-i})\ge \hat{u}_i^{v_i}(\mathbf{x}(a_i^*, \mathbf{a}_{-i}))$. Therefore,
\begin{align*}
\sum_i u_i^{v_i}(a_i^*, \mathbf{a}_{-i}) &\ge \sum_i \hat{u}_i^{v_i}(a_i^*, \mathbf{a}_{-i})\ge \lambda \widehat\OPT - \mu_1 \sum_i p_i(\mathbf{a}) - \mu_2 \sum_i W_i(a_i, X(\mathbf{a}))\\ 
&\ge\frac{\lambda}{2} \OPT - \mu_1 \sum_i p_i(\mathbf{a})- \mu_2 \sum_i W_i(a_i, X(\mathbf{a}))\enspace,
\end{align*}
where the last inequality follows from Lemma~\ref{lemma:OPT_relation}.
\end{proof}

Next, we will show that pointwise no-overbidding indeed implies the no-overbidding assumption~\eqref{eq:generic_no-overbidding}:

Using the pontwise no-overbidding assumption $v_i(\mathbf{x})\ge p_i$, we know that $u_i^{v_i}(\mathbf{x},p_i)\ge v_i(\mathbf{x})-p_i$. From here, $W_i(a_i, \mathbf{x}) \le v_i(\mathbf{x})\le u_i^{v_i}(\mathbf{a})+p_i(\mathbf{a})$, so we can conclude that
\[
\mathbb{E}_{\mathbf{a}}[W_i(a_i, X(\mathbf{a}))]\le 
\mathbb{E}_{\mathbf{a}}[u_i^{v_i}(\mathbf{a}) + p_i(\mathbf{a})]\enspace.
\]
Theorem~\ref{thm:weak_smoothness} now completes the proof.
\end{proof}

Using that the second-price auction is weakly $(1, 0 ,1)$-smooth with respect to quasilinear utilities, we immediately get that its price of anarchy is also constant in the risk-averse setting.

\begin{corollary}
Under normalized risk-averse utilities, the second-price auction has a constant price of anarchy for correlated and Bayes-Nash equilibria with pointwise no-overbidding.
\end{corollary}

\section{Budget Constraints}\label{sec:budgets}
The techniques and results so far have striking similarities to settings with budget constraints, where players do not have quasilinear preferences already in the risk neutral case. As it turns out, under very mild additional assumptions, we can also add (a generalized form of) hard budget constraints to our consideration. 

We now assume that types are pairs $\theta_i = (v_i, B_i)$, where $B_i\colon \mathcal{X}_i\to \mathbb{R}^+$ is an outcome-dependent budget function. Depending on which outcome is achieved, the agent may have different amounts of liquidity. We assume that there is a normalized risk-averse utility function $u_i^{v_i}$ such that for a player of type $\theta_i = (v_i, B_i)$
\[
u_i^{\theta_i}(\mathbf{x}, p_i) = \begin{cases} u_i^{v_i}(\mathbf{x}, p_i) & \text{ if $p_i \leq B_i(\mathbf{x})$} \\
- \infty & \text{ otherwise} \end{cases} \enspace.
\]

In the budgeted setting, one cannot hope to achieve full welfare. This is due to low budget participants not being able to maximize their contribution. Therefore, we will replace $\OPT(\boldsymbol\theta)$ in the price-of-anarchy and smoothness definition by the optimal \emph{effective} or \emph{liquid} welfare, given as $\max_{\mathbf{x}, \mathbf{p}} \sum_i \min\{u_i^{\theta_i}(\mathbf{x}, p_i) + p_i, B_i(\mathbf{x})\}$. This benchmark, introduced in~\cite{DobzinskiL14}, reflects that players with low budgets cannot be expected to be effective at maximizing their own value.

The effect of budgets on efficiency in the risk neutral case was already studied in~\cite{SyrgkanisT13}, where the authors, in order to be able to prove efficiency bounds, introduced the notion of a \emph{conservatively smooth mechanism} that has the following additional assumption on the smoothness deviations:
\begin{equation}\label{eq:conservatively_smooth}
\max_{\mathbf{a}_{-i}} p_i(a_i^*(\mathbf{v}, a_i), \mathbf{a}_{-i})\le \max_{\mathbf{x}} v_i(\mathbf{x})\enspace.
\end{equation}
Conservatively smooth mechanisms are then shown to allow the budgeted scenario without any further loss of efficiency.
Note that~(\ref{eq:conservatively_smooth}) is a weaker assumption than the Condition~(\ref{eq:nonnegutility}) we ask for. Therefore, we can easily extend our results for risk-averse bidders to the budgeted setting.

Our main result is that if the type space is chosen in a way that taking the pointwise minimum of a valuation function and a budget function yields again a feasible valuation function, meaning that we stay within the ``permitted'' valuation space when applying the budget costraints, then the price-of-anarchy guarantee is again preserved. The valuation space being closed under capping is a crucial requirement both for our result and the result in~\cite{SyrgkanisT13}.

\begin{theorem}\label{thm:budgets_risk}
If a mechanism is $(\lambda, \mu)$-smooth w.r.t. quasilinear utility functions $(\hat{u}_i^{v_i})_{i\in N, v_i\in\mathcal{V}_i}$, its valuation space is closed under capping with budget functions, the actions in the support of the smoothness deviations satisfy $\hat{u}_i(a_i^*, \mathbf{a}_{-i})\geq0, \forall \mathbf{a}_{-i}, \forall i$, then the social welfare at any Correlated Equilibrium and at any Bayes-Nash Equilibrium is at least $\frac{\lambda}{2\cdot\max\{1,\mu\}}$ of the expected maximum effective welfare even in the presence of risk-averse bidders.
\end{theorem}

As before, we prove a lemma connecting smoothness with respect to quasilinear utilities to smoothness with respect to risk-averse ones.

\begin{lemma}\label{lemma:budgets}
If a mechanism is $(\lambda, \mu)$- smooth w.r.t. quasilinear utility functions $(\hat{u}_i^{v_i})_{i\in N, v_i\in\mathcal{V}_i}$, its valuation space is closed under capping with the budget functions, and the actions in the support of the smoothness deviations satisfy $\hat{u}_i(a_i^*, \mathbf{a}_{-i})\geq0, \forall \mathbf{a}_{-i}, \forall i$, then the mechanism is $(\lambda/2, \mu)$-smooth with respect to risk-averse budgeted utility functions $(u_i^{\theta_i})_{\theta_i\in\Theta_i, i\in N}$.
\end{lemma}
\begin{proof}
We start from an arbitrary action profile $\mathbf{a}$ and keep in mind that the risk-averse budgeted utility function $u_i^{\theta_i}$ has type $\theta_i=(v_i, B_i)$. By $\hat{u}^{\bar{v}_i}$, we denote the quasilinear utility of player $i$ with the capped valuation function $\bar{v}_i$. Formally,
\[
\hat{u}_i^{\bar{v}_i}(\mathbf{x},p_i)=\bar{v}_i(\mathbf{x})-p_i=\min\{v_i(\mathbf{x}), B_i(\mathbf{x})\}-p_i\enspace.
\]
Since the valuation space is closed under capping with the budget function, we can find smoothness deviations $a_i^*(\bar{\mathbf{v}}, a_i)$ s.t. $\hat{u}_i^{\bar{v}_i}(X(a_i^*, \mathbf{a}_{-i}),p_i(a_i^*, \mathbf{a}_{-i}))=\bar{v}_i(X(a_i^*, \mathbf{a}_{-i}))-p_i(a_i^*, \mathbf{a}_{-i})\geq0, \forall \mathbf{a}_{-i}$ and therefore $u_i^{\bar{v}_i}(a_i^*, \mathbf{a}_{-i})\ge \hat{u}_i^{\bar{v}_i}(a_i^*, \mathbf{a}_{-i})$. It follows that
\begin{align*}
\sum_i u_i^{\theta_i}(a_i^*(\bar{\mathbf{v}}, a_i), \mathbf{a}_{-i})&= \sum_i u_i^{v_i}(a_i^*(\bar{\mathbf{v}}, a_i), \mathbf{a}_{-i})
\ge \sum_i u_i^{\bar{v}_i}(a_i^*, \mathbf{a}_{-i})
\ge \sum_i \hat{u}_i^{\bar{v}_i}(a_i^*, \mathbf{a}_{-i})\\
&\ge \lambda\cdot \widehat{\OPT}_{\bar{\mathbf{v}}} - \mu\cdot\sum_i p_i(\mathbf{a})\ge\frac{\lambda}{2}\cdot \OPT_{\bar{\mathbf{v}}} - \mu\cdot\sum_i p_i(\mathbf{a}) \enspace,
\end{align*}
where the first equality holds because the deviations are such that the payments are below the budgets, the first inequality because 
\[
u_i^{v_i}(\mathbf{a})=h\left(v_i(\mathbf{x}(\mathbf{a})-p_i(\mathbf{a})\right)\ge h(\min\{v_i(\mathbf{x}(\mathbf{a})), B_i\}-p_i(\mathbf{a})) = u_i^{\bar{v}_i}(\mathbf{a}), \forall \mathbf{a}\enspace,
\]
and the third because the valuation space is closed under capping. 
\end{proof}

The generality of Theorem~\ref{thm:smoothness_CCE_and_BNE} allows us to now obtain Theorem~\ref{thm:budgets_risk}. Note that $\OPT_{\bar{\mathbf{v}}}$, where $\bar{\mathbf{v}}$ is the vector of capped valuation functions, indeed aligns correctly with the effective welfare benchmark.

\section{Missing Details from Section~\ref{sec:negative-results}}\label{sec:calculations}
\subsection{Calculations}\label{subsec:calculations}
\begin{align*}
\mathbb{E}[u_3(b_3', \mathbf{b}_{-3})] &\leq (1 - \epsilon)\cdot(\frac{1}{3}\cdot\ln M/2 - b_3') + \epsilon\cdot(- (16\cdot\frac{1}{3}\cdot\ln M/2 -1)\cdot M^2 \cdot b_3')\\
 &< \frac{1}{3}\cdot\ln M/2 - b_3' - \frac{1}{M^2}\cdot \big(16\cdot\frac{1}{3}\cdot\ln M/2 -1\big)\cdot M^2 \cdot b_3'\\
&=\frac{1}{3}\cdot\ln M/2 - b_3' (1 + 16\cdot\frac{1}{3}\cdot\ln M/2 -1)\\
&< \frac{1}{3}\ln M/2 - \frac{1}{16}\cdot 16\cdot\frac{1}{3}\ln M/2 = 0\enspace.
\end{align*}

\begin{align*}\label{eq:bid_estimate}
\beta(v_1)&\ge \int_{1/2}^{v_1} \frac{2(1-\frac{M-1}{M^2})(e^t-1)}{2(t-\frac{1}{2})(1-\frac{M-1}{M^2})+ (1-2(t-\frac{1}{2})(1-\frac{M-1}{M^2}))\cdot e^t} dt\\
&>\int_{1/2}^{v_1}\frac{\frac{3}{2}(e^t-1)}{2t-1 + 2\cdot e^t} dt=\frac{3}{4}\int_{1/2}^{v_1}\frac{e^t-1}{e^t+t-\frac{1}{2}}dt\\
&\ge \frac{3}{4}\int_{1/2}^{v_1} \left(1-\frac{1}{\sqrt{e}}\right) dt=\frac{3}{4}\left(1-\frac{1}{\sqrt{e}}\right)\left( v_1 - \frac{1}{2} \right) >\frac{1}{4} \left( v_1 - \frac{1}{2} \right) \enspace.
\end{align*}

\begin{align*}
\mathbb{E}[u_3(b_3', \mathbf{b}_{-3})] &\le\Pr[\beta(v_1) \leq b_3']\cdot(v_3-b_3')-\Pr[\beta(v_1) > b_3']\cdot 32\cdot v_3\cdot b_3' \\
& <2\cdot 4b_3'(v_3-b_3')-(1-2\cdot 4b_3')\cdot 32\cdot v_3\cdot b_3'\\
&=8b_3'v_3-8(b_3')^2-32b_3'v_3+8\cdot32v_3(b_3')^2 \\
&=8b_3'\big(-3v_3+b_3'\big(32v_3-1\big)\big)\le8b_3'\left(-3v_3 + 2v_3 - \frac{1}{16}\right)\\
&=8b_3'\left(-v_3-\frac{1}{16}\right)<0\enspace.
\end{align*}

\subsection{All-Pay Auction with Limited Risk-Aversion}\label{sec:bounded_slope}
\begin{theorem}
In an all-pay auction with risk-averse players whose utilities are of the form $h(v_i(\mathbf{x})-p_i)$, where $h$ is a concave function s.t. $h(x)=C\cdot x$ for $x<0$, $C\ge1$ constant, the Price of Anarchy is at most $4(C+1)$.
\end{theorem}
\begin{proof}
We use the following smoothness deviation: the highest value player with value $v_h$ deviates to $\frac{1}{2}v_h$ and everybody else to $0$. Now, it is easy to see that the following inequality holds independent of whether the highest value player obtains the item or not
\[
u_h^{v_h}(\frac{v_h}{2}, \mathbf{a}_{-i})\ge \frac{1}{2}v_h - (C+1)\max_{i\neq h}a_i\ge \frac{1}{2}\widehat{\OPT} - (C+1)\sum_i a_i\enspace,
\]
so then
\[
\sum_i u_i^{v_i}(a_i^*, \mathbf{a}_{-i}) \ge \frac{1}{2}\widehat{\OPT} - (C+1)\sum_i p_i(\mathbf{a})\ge \frac{1}{4}\OPT - (C+1)\sum_i p_i(\mathbf{a})\enspace.
\]
The claim follows by applying Theorem~\ref{thm:smoothness_CCE_and_BNE}.
\end{proof}

\subsection{Symmetric BNE of All-Pay Auction in the Quasilinear Setting}\label{sec:quasilinear_allpay}
\begin{claim}
In a symmetric BNE of the all-pay auction in the quasilinear setting, all bids are bounded by the expected value of a player.
\end{claim}
\begin{proof}
Due to symmetry, it is enough to argue about the first player. Let $\beta$ denote the equilibrium bidding function. We fix player 1's value $v_1=x$ and consider his expected utility for bidding $y$:
\begin{align*}
\mathbb{E}[u_1^{x}(b_1=y, b_2=\beta(v_2))] &= \Pr[\beta(v_2)<y]\cdot(x-y) + \Pr[\beta(v_2)\ge y]\cdot (-y)\\
&= \Pr[\beta(v_2)<y]\cdot x - y = F(\beta^{-1}(y))\cdot x -y\enspace.
\end{align*}
By taking the derivative and setting it to zero, we arrive at 
\[
\frac{f(\beta^{-1}(b))}{\beta'(\beta^{-1}(b))}\cdot x -1 =0\enspace,
\]
so 
\[
\beta'(x)=x\cdot f(x)\enspace.
\]
Now it is obvious that 
\[
\beta(x)=\int_0^x t\cdot f(t) \le \mathbb{E}[v_1]\enspace.
\]
\end{proof}

\section{Proof of Observation~\ref{lemma:observation}}\label{sec:observation}
\begin{proof}
Consider two bidders that both have a valuation of $1$. If they both bid $1$ with probability $\frac{1}{2}$ and $0$ with the remaining probability, but in a correlated manner, such that always just one player submits a non-zero bid -- they will be in an equilibrium. Let us now calculate the utilities:
\[
u_i(\mathbf{b}) = \mathbb{E}_{\mathbf{a}\sim\mathbf{b}}[\hat{u}_i(\mathbf{a})] - \sqrt{\mathbb{E}[\hat{u}_i^2(\mathbf{a})] - (\mathbb{E}[\hat{u}_i(\mathbf{a})])^2} = \frac{1}{2} - \sqrt{\frac{1}{2} - \Big(\frac{1}{2}\Big)^2} = \frac{1}{2}-\frac{1}{2}=0\enspace.
\]
Since the payments are also $0$, the social welfare in this equilibrium is $0$, meaning that the price of anarchy is unbounded.
\end{proof}

\end{document}